\documentclass[conference]{IEEEtran}
\usepackage[dvips]{graphicx}
\usepackage{subfigure}
\usepackage{psfrag}
\usepackage{stmaryrd}
\usepackage{dsfont}
\usepackage{amsmath,amsthm}
\usepackage[latin1]{inputenc}
\usepackage{times}
\usepackage{cite}
\usepackage[active]{srcltx}
\usepackage{color}
\usepackage{url}

\newcommand{\word}[1]{{\boldsymbol #1}}

\newtheorem{theorem}{Theorem}
\newtheorem{lemma}[theorem]{Lemma}

\IEEEoverridecommandlockouts 

\begin{document}

\title{Bilayer LDPC Convolutional Codes for Half-Duplex Relay Channels}

\author{\IEEEauthorblockN{Zhongwei  Si,  Ragnar  Thobaben, and  Mikael
    Skoglund \thanks{This  work was supported in  part  by the  European Community's  Seventh
    Framework  Programme under  grant  agreement  no 257626 (ACROPOLIS) and the Swedish Foundation for Strategic Research.}}
  \IEEEauthorblockA{School of Electrical Engineering and
    ACCESS Linnaeus Center\\
    Royal Institute of Technology (KTH), Stockholm, Sweden} }

\maketitle

\begin{abstract}
In this paper we present regular bilayer LDPC convolutional codes
for half-duplex relay channels. For the binary erasure relay
channel, we prove that the proposed code construction achieves the
capacities for the source-relay link and the source-destination
link provided that the channel conditions are known when designing
the code. Meanwhile, this code enables the highest transmission
rate with decode-and-forward relaying. In addition, its regular
degree distributions can easily be computed from the channel
parameters, which significantly simplifies the code optimization.
Numerical results are provided for both binary erasure channels
(BEC) and AWGN channels. In BECs, we can observe that the gaps
between the decoding thresholds and the Shannon limits are
impressively small. In AWGN channels, the bilayer LDPC
convolutional code clearly outperforms its block code counterpart
in terms of bit error rate.

\end{abstract}

\section{Introduction}

The relay channel was introduced in 1971 when van der Meulen \cite
{Meu71} proposed a channel model consisting of one source, one
relay, and one destination. The relay aids the communication
between the source and the destination so that increased
robustness, higher transmission efficiency, and/or larger coverage
range can be achieved. As smallest but fundamental unit of large
network topologies, the relay channel has been extensively studied
focusing on both theoretical and implementation aspects.

Decode-and-forward (DF) relaying is the most researched protocol
for relay channels. In particular, the design of distributed
channel codes has attracted considerable attention. The concept of
distributed Turbo coding (DTC) was proposed in \cite{ZV03}, which
offered a new fashion of distributed code design. Low-density
parity-check (LDPC) codes were considered for distributed coding
for example in \cite{CBSA07}, \cite{HD07} and \cite{RY07}.
Different approaches were presented to optimize LDPC codes for
given channel conditions. For LDPC block codes, an irregular
degree distribution needs to be derived to match a given channel.
For a variety of channel conditions, extensive re-optimization is
required. This leads to a high complexity for code adaptation and
may not be feasible in practice.

In this paper we propose to use LDPC convolutional codes for
distributed channel coding in relay networks. LDPC convolutional
codes were first proposed in \cite{FZ99} as a time-varying
periodic LDPC code variation. Then the idea was further developed
in, e.g., \cite{PFSLZC08}, \cite{LSCZ10}. Recently, it has been
proven analytically in \cite{KRU10} that the belief-propagation
(BP) decoding threshold of an LDPC convolutional code achieves the
optimal maximum \emph{a posteriori} probability (MAP) threshold of
the corresponding LDPC block code with the same variable and check
degrees. This code in turn approaches the capacity as the node
degrees increase. Furthermore, regular LDPC convolutional codes
allow us to avoid complicated re-optimization of the degree
distributions for varying channel conditions. Meanwhile, LDPC
convolutional codes enable recursive encoding and sliding-window
decoding \cite{LSCZ10}, which dispels the concerns over complexity
and delay. Motivated by the good properties of LDPC convolutional
codes, we consider in this paper the design of bilayer LDPC
convolutional codes for the relay channel. A similar code
construction was proposed in \cite{RUAS10} for the wiretap
channel. A protograph-based bilayer code was proposed in
\cite{NND10} which applies the concept of bilayer-lengthened
codes. In contrast to \cite{NND10} we present bilayer expurgated
codes \cite{RY07} in this paper.

In the following, we will discuss the construction of bilayer LDPC
convolutional codes for given relay channels. We will prove
analytically that the proposed bilayer code is capable of
achieving the highest rate with DF relaying in binary erasure
channels (BEC). Moreover, the regularity of degree distributions
significantly simplifies the code optimization. Numerical results
are provided to verify the theoretical analysis.

\section{Preliminaries}
\label{sec:Preliminary}

In this section, firstly we introduce the transmission model we
use throughout the paper. Then we briefly review the coding
strategy which leads to the highest achievable rate \cite{KSA03}
with DF relaying.  The construction of bilayer codes \cite{RY07}
is described as a practical realization of the coding strategy.

\subsection{System Model}
In this paper, we restrict ourself to the three-node relay channel
which is composed of one source, one relay, and one destination.
The source ($S$) intends to transmit its information to the
destination ($D$) while the relay ($R$) provides assistance.

The system model is shown in Figure~\ref{fig:Transmod}. Due to
practical constraints the relay works in a half-duplex mode, which
means it cannot transmit and receive at the same time or the same
frequency. This implies that the transmission from the source to
the destination is carried out in two phases. In the first phase,
the source broadcasts while the relay and the destination listen.
In the second phase, the relay transmits to the destination while
the source keeps silent. We assume the transmissions on the three
links to be orthogonal.

\begin{figure}[htb]
    \psfrag{B}[][]{\scriptsize $\word{B}$}
    \psfrag{D}[][]{\scriptsize $\hat{\word{B}}$}
    \psfrag{Xr}[][]{\scriptsize $\word{X}\!_R$}
    \psfrag{Xs}[][]{\scriptsize $\word{X}\!_S$}
    \psfrag{Y1}[][]{\scriptsize $\word{Y}\!_{SR}$}
    \psfrag{Y2}[][]{\scriptsize $\word{Y}\!_{SD}$}
    \psfrag{Y3}[][]{\scriptsize $\word{Y}\!_{RD}$}
  \centering
  \includegraphics[width=.45\textwidth]{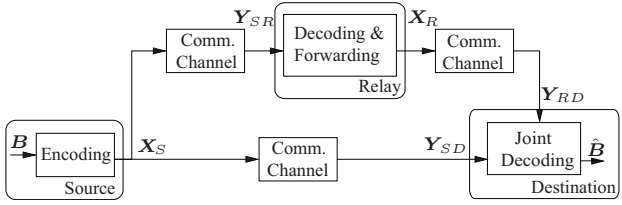}
  \vspace{-3mm}
  \caption{Transmission model.}
  \vspace{-2mm}
  \label{fig:Transmod}
\end{figure}
In the following we use  $X_i$, $i\in\{S,R\}$, to denote  the BPSK
modulated signals which  are  transmitted from  the source and the
relay, and   we  use   $Y_{i,j}$,  $i\in\{S,R\}$, $j\in\{R,D\}$,
for  the channel observations  of the three  links. We use
$C_{ij}=\emph{I}(X_i;Y_{ij})$, $i\in\{S,R\}$, $j\in\{R,D\}$ to
denote the capacity of  each link constrained  to the BPSK
modulation. In this paper we assume that perfect channel-state
information (CSI) is available for code construction.

\subsection{Achievable Rate}

The highest transmission rate using decode-and-forward protocol
for the half-duplex relay channel with orthogonal receive
components is given as \cite{KSA03}
\begin{equation}
\label{AchiRate}
\begin{split}
  R_{DF}= \sup_{\alpha,0\leq \alpha\leq 1} \min (\alpha \emph{I}(X_S;Y_{SR}&),\\
  \alpha  \emph{I}(X_S;Y_{SD})+(1-&\alpha)
  \emph{I}(X_R;Y_{RD}))
\end{split}
\end{equation}
where $\alpha$ is  the fraction  of channel  uses in  the first
phase, and $(1-\alpha)$ is the fraction of channel uses in the
second phase.

To achieve $R_{DF}$, in the first phase the source employs a
capacity-achieving code for the source-relay link which guarantees
successful decoding at the relay. This code may not be decodable
at the destination due to the poorer channel condition on the
source-destination link. Therefore, in the second phase the relay
forwards additional bits to the destination in order to construct
an overall lower rate code which is capacity-achieving for the
source-destination link. A practical implementation of the idea is
presented in the following.

\subsection{Bilayer LDPC Block Codes for Relay Channels}
\label{sec:BlockBiLDPC}

The construction of bilayer LDPC block codes \cite{RY07} is
realized in two steps corresponding to the two transmission
phases.

In the first phase, $K_1$ information bits $\word{B}$ are encoded
by a length-$N_1$ codeword $\word{X}_S$ through a rate-$R_1$ LDPC
code $\mathcal{C}_1$ (i.e., $K_1=N_1 \cdot R_1$) with the check
matrix $\word{H}_S$ and transmitted. At the end of the first
phase, the relay decodes $\mathcal{C}_1$, using the check matrix
$\word{H}_S$, and recovers $\word{X}_S$.

At the destination, additional $K_2$ bits are needed for
successfully decoding $\word{X}_S$:
\begin{equation*}
K_2=N_1(\emph{I}(X_S;Y_{SR})-\emph{I}(X_S;Y_{SD})).
\end{equation*}
Therefore, in the second phase the relay generates $K_2$ new bits
(syndrome, $\word{S}$) using the check matrix $\word{H}_R$. These
$K_2$ syndrome bits are transmitted to the destination via a
channel encoder $\mathcal{C}_2$ of rate $R_2$ using $N_2$ channel
uses, i.e., $K_2=N_2 \cdot R_2$. To simplify the discussion, we
assume these syndrome bits are perfectly known at the destination
after decoding $\mathcal{C}_2$. Then the overall code
$\mathcal{C}$ is described by the stacked check matrix $\word{H}$,
and we have
\begin{equation*}
\word{H}\word{X}_S=\left[\begin{array}{c} \word{H}_S\\\word{H}_R\end{array}\right]\word{X}_S=\left[\begin{array}{c} \word{0}\\\word{S}\end{array}\right].
\end{equation*}
That is, at the destination ($N_1-K_1$) zero check equations and
$K_2$ non-zero check equations need to be satisfied in the
decoding. The Tanner graph of a bilayer code example is plotted in
Figure~\ref{fig:BilayerLDPC}.

\begin{figure}[htb]
\vspace{-2mm}
   \centering
  \includegraphics[width=.25\textwidth]{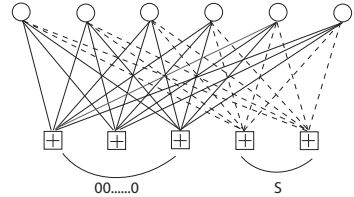}
  \vspace{-2mm}
  \caption{The Tanner graph of a bilayer LDPC code. Circles represent variable nodes, and squares indicate check nodes.
  The solid lines correspond to the edges in $\word{H}_S$, and the dashed lines show the connections determined by $\word{H}_R$.}
  \vspace{-2mm}
  \label{fig:BilayerLDPC}
\end{figure}

To achieve the optimal performance, the design of $\word{H}_S$ and
$\word{H}$ needs to guarantee that $\mathcal{C}_1$ and
$\mathcal{C}$ are simultaneously capacity-achieving for the
source-relay link and the source-destination link respectively.
The authors of \cite{RY07} approached this target by applying
irregular LDPC block codes. Consequently, re-optimization is
required for every given channel, which results in high complexity
and infeasibility. In the next section we will show how this goal
can be achieved by using regular LDPC convolutional codes leading
to significantly reduced optimization overhead.

If the channel codes $\mathcal{C}_1$ and $\mathcal{C}_2$ are both
capacity-achieving, i.e., $R_1=C_{SR}$ and $R_2=C_{RD}$, then the
achievable rate in (\ref{AchiRate}) is maximized by \vspace{-1mm}
\begin{equation}
\label{fraction}
\alpha = \frac{N_1}{N_1+N_2} = \frac{C_{RD}}{C_{RD}+C_{SR}-C_{SD}}.
\end{equation}
Later in this paper we will prove that in BECs $R_{DF}$ can be
achieved by applying bilayer LDPC convolutional codes.

\section{Bilayer LDPC Convolutional Codes}
\label{sec:BilayerLDPCCC}

\subsection{LDPC Convolutional Codes}

A regular $(l,r)$ time-varying binary LDPC convolutional code can
be defined by a syndrome former matrix \cite{LSCZ10}
\begin{equation*}
\word{H}^T=\left[\begin{array}
{c c c c c}
\ddots & & \ddots & &\\
\word{H}_0^T(1) & \ldots & \word{H}_w^T(1+w) &  & \\
  & \ddots & & \ddots & \\
  & & \word{H}_0^T(t) & \ldots & \word{H}_w^T(t+w)\\
  & & \ddots & & \ddots \\
\end{array}\right],
\end{equation*}
where $l$ is the variable degree and $r$ is the check degree. We
assume that at each time instant $t$ ($t=1,2,...,L$) the number of
variable nodes is $M$. Then each submatrix $\word{H}_i^T(t+i)$ is
a $M\times(Ml/r)$ binary matrix. The largest $i$ such that
$\word{H}_i^T(t+i)$ is nonzero for some $t$ is called the syndrome
former memory $w$. The matrix $\word{H}^T$ is sparse.

There are many variations of LDPC convolutional codes in the
literature. In this paper, we denote an LDPC convolutional code by
four parameters $\{l,r,L,w\}$. The memory constraint $w$ can be
any non-negative integer. We assume that each of the $l$ edges of
a variable node at time $t$ uniformly and independently connects
to the check nodes in the time range $[t,...,t+w]$. More
precisely, for each variable node at time $t$, one can define a
\emph{type} $\mathcal{M}_t$\footnote{Index of the variable node is
omitted for the ease of notation.} \cite{KRU10} which is a
$w$-tuple of non-negative integers,
$\mathcal{M}_t=(m_{t,t},...,m_{t,t+j},...,m_{t,t+w})$,
$j\in[0,w]$, and $\sum_{j} m_{t,t+j}=l$. The element $m_{t,t+j}$
indicates that there are $m_{t,t+j}$ edges connecting the
designated variable node at time $t$ and the check nodes at time
$t+j$. For each variable node, $\mathcal{M}_t$ is uniformly and
independently chosen from all possible types. It has been stated
in \cite{KRU10} that the $\{l,r,L,w\}$ code ensemble is capacity
achieving and easier to analyze. However, experimentally it shows
a worse trade-off between rate, threshold and block length.

Another variant, the $\{l,r,L\}$ ensemble, can be considered as a
special case of the more general code ensemble mentioned above.
For this ensemble, the memory length $w$ always equals $l-1$.
Exactly one of the $l$ outgoing edges of each variable node at
time $t$ is connected to one check node at position
$[t,...,t+(l-1)]$, i.e., $m_{t,t+j}\!=\!1$ for all $j\in[0,l-1]$.
We observe through experiments that this type of ensemble provides
good performance with moderate $M$ and $L$ when $l\geq 3$.

In this paper, we use the $\{l,r,L,w\}$ ensemble for theoretical
analysis while employing the $\{l,r,L\}$\! ensemble in
simulations.

\subsection{Bilayer LDPC Convolutional Codes for Relay Channels}

Firstly, we define the structure of a bilayer LDPC convolutional
code. We assume the number of variable nodes to be $N=M \cdot L$.
The connections between the $N$ variable nodes and the check nodes
in the first (second) layer are determined by the ensemble
$\{l_1,r_1,L,w_1\}$ ($\{l_2,r_2,L,w_2\}$). If $w_1=w_2$, we denote
the bilayer code by $\{l_1,l_2,r_1,r_2,L,w\}$. Note that only the
edges belonging to the same layer are connected to one check node.
The structure of the overall check matrix is illustrated in
Figure~\ref{fig:LDPCCBilayer}.

\begin{figure}[htb]
\vspace{-3mm}
\psfrag{D}[][]{$\word{H}^T$}
\psfrag{A}[][]{ $\word{H}_{S}^T$}
\psfrag{B}[][]{$\word{H}_{R}^T$}
\psfrag{E}[][]{ $\ddots$}
   \centering
  \includegraphics[width=.35\textwidth]{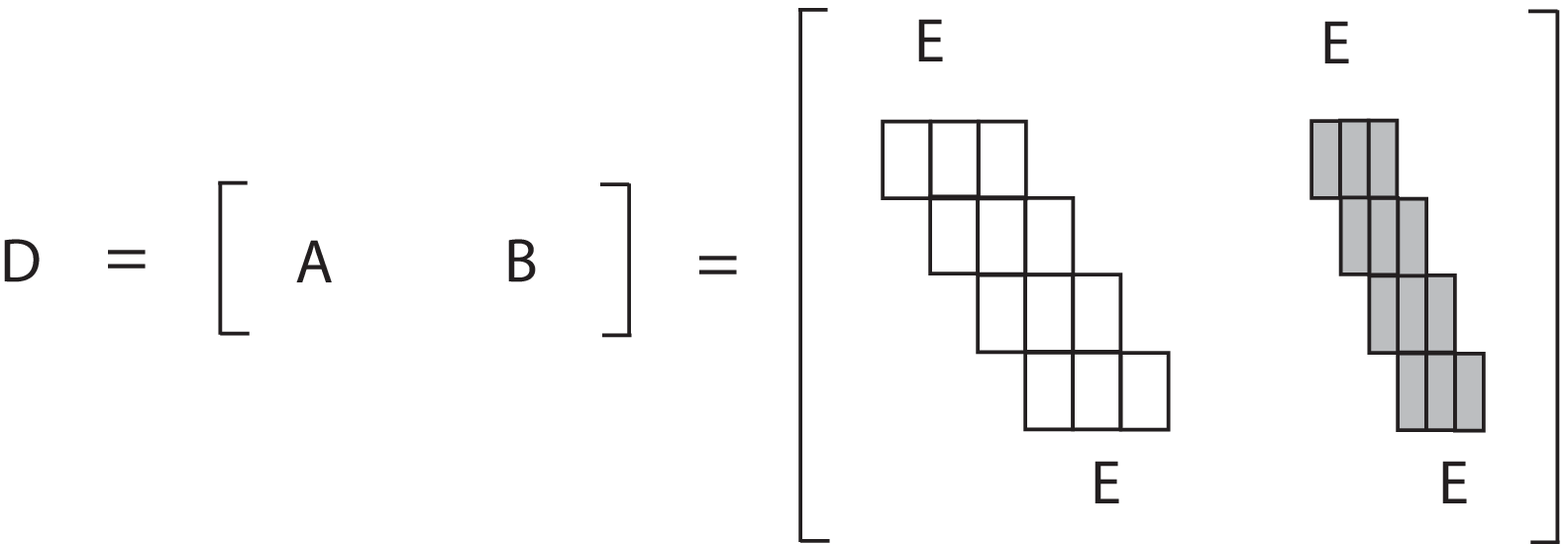}
  \vspace{-1mm}
  \caption{Overall check matrix of a bilayer LDPC convolutional code.
  The white blocks correspond to the non-zero submatrices in the first layer, and the grey blocks are for those submatrices in the second layer.}
  \vspace{-3mm}
  \label{fig:LDPCCBilayer}
\end{figure}

The protocol for transmitting a bilayer LDPC convolutional code
for the relay channel is similar to the strategy we explained in
Section~\ref{sec:BlockBiLDPC}. The information bits from the
source are encoded by the single-layer code $\{l_1,r_1,L,w_1\}$
and broadcasted in the first phase. After successful decoding, the
relay generates the syndrome bits using $\{l_2,r_2,L,w_2\}$. These
syndrome bits are transmitted to the destination under perfect
protection by another channel code in the second phase. The
destination decodes the overall code by considering the zero check
equations in the first layer and the non-zero check equations in
the second layer.

\subsection{Analysis for Binary Erasure Channels}

It has been shown in \cite{KRU10} that the $\{l,r,L,w\}$ ensemble
with infinite $M$ has the following properties in a binary erasure
channel: for the rate of the code $R$
 \vspace{-1mm}
\begin{equation}
\lim_{w\rightarrow\infty}\lim_{L\rightarrow\infty}R(l,r,L,w)=1-\frac{l}{r},
\label{Rate}
\end{equation}
and for the decoding threshold
\begin{equation*}
\lim_{w\rightarrow\infty}\!\lim_{L\rightarrow\infty}\!\epsilon^{B\!P}(l,r,L,w)\!=\!\!\lim_{L\rightarrow\infty}\!\epsilon^{M\!A\!P}(l,r,L,w)\!=\!\epsilon^{M\!A\!P}(l,r),
\end{equation*}
where $\epsilon^{BP}$ and $\epsilon^{MAP}$ are respectively the BP
threshold and the MAP threshold for decoding. If we increase the
degrees of the nodes, its decoding threshold approaches the
Shannon limit $\epsilon_{Sh}=1-R$,
 \vspace{-1mm}
\begin{equation}
\lim_{r\rightarrow\infty}\lim_{w\rightarrow\infty}\lim_{L\rightarrow\infty}\epsilon^{BP}(l,r,L,w)=\epsilon_{Sh}.
\label{Capacity}
\end{equation}

In the following, we will show in Theorem 2 that the bilayer LDPC
convolutional code $\{l_1,l_2,r,r,L,w\}$ achieves the same Shannon
limit as the standard single-layer ensemble $\{l_1+l_2,r,L,w\}$
\cite{RUAS10}. As a preparation for the theorem, we introduce the
following lemma.

\begin{lemma}
If $M$, $L$ and $w$ go to infinity in this order, the density
evolution of a single-layer LDPC convolutional code $\{l,r,L,w\}$
in a binary erasure channel can be written as
\begin{equation*}
p^{(i)}=\epsilon (q^{(i-1)})^{l-1} \hspace{3mm} \mathrm{and} \hspace{3mm} q^{(i)}=1-(1-p^{(i)})^{r-1},
\end{equation*}
where $p^{(i)}$ ($q^{(i)}$) is the erasure probability from a
variable (check) node to a check (variable) node in the $i$-th
iteration, and $\epsilon$ is the erasure probability of the
channel.
\end{lemma}

\begin{proof}
In the following we refer to the check nodes connected to a given
variable node as the active check nodes for that variable node. We
use $p_{t,t+j}^{m,(i)}$ to denote the probability that the message
from a given variable node at time $t$ to the $m$-th active check
node at time $t+j$ in decoding iteration $i$ is erased. In the
first iteration, $p_{t,t+j}^{m,(1)}=\epsilon$ for all $m$ and $j$.
We use $q_{t+j,t}^{n,(i)}$ to represent the probability that the
message from the $n$-th active check node at $t+j$ to the given
variable node at $t$ is erased. Then we have
\begin{equation}
\label{Lem1}
p_{t,t+j}^{m,(i)}=\epsilon\prod_{k=0,k\neq j}^{w} \left(\prod_{n=1}^{m_{t,t+k}}q_{t+k,t}^{n,(i-1)}\right) \cdot \prod_{v=1,v\neq m}^{m_{t,t+j}}q_{t+j,t}^{v,(i-1)}.
\end{equation}
If $M\rightarrow\infty$, the messages from different nodes at the
same time instant behave identically \cite{LSCZ10}. Then
(\ref{Lem1}) reduces to
\begin{equation}
\label{Lem2}
p_{t,t+j}^{(i)}=\epsilon\prod_{k=0,k\neq j}^{w} (q_{t+k,t}^{(i-1)})^{m_{t,t+k}} \cdot (q_{t+j,t}^{(i-1)})^{m_{t,t+j}-1}.
\end{equation}

The messages from nodes at different time instants can behave
differently and are usually tracked separately. However, if we
have $L\rightarrow\infty$, the effect of boundaries caused by the
initialization and the termination of the code vanishes. We can
then consider the code asymptotically regular \cite{LFZC09}. The
message updating is averaged over $w+1$ time instants. Therefore,
if $w\rightarrow\infty$, the messages from the nodes at different
time instants have asymptotically identical distribution.
Eventually, (\ref{Lem2}) is simplified to
\vspace{-0.5mm}
\begin{equation*}
p^{(i)}=\epsilon (q^{(i-1)})^{l-1}.
\end{equation*}
\vspace{-5mm}

Similarly, we also obtain \vspace{-1mm}
\begin{equation*}
q^{(i)}=1-(1-p^{(i)})^{r-1}.
\vspace{-5mm}
\end{equation*}
\end{proof}

Now we look at the relation between the bilayer LDPC convolutional
code ensemble $\{l_1,l_2,r,r,L,w\}$ and the standard single-layer
ensemble $\{l_1+l_2,r,L,w\}$.

\begin{theorem}{\cite{RUAS10}}
We denote a bilayer LDPC convolutional code of length $N=M \cdot
L$ by $\{l_1,l_2,r_1,r_2,L,w\}$, where $l_1$ and $l_2$ are
respectively the variable degrees of the two layers, $r_1$, $r_2$
are the check degrees of the two layers, and $w$ is the common
memory constraint. If we assume the two layers take the same check
degree, i.e., $r_1=r_2=r$, then the bilayer LDPC convolutional
code $\{l_1,l_2,r,r,L,w\}$ approaches the same Shannon limit as
the single-layer LDPC convolutional code $\{l_1+l_2,r,L,w\}$.
\end{theorem}

\begin{proof}
For the completeness of the proof, we repeat the derivation of the
BP decoding threshold which was previously given in \cite{RUAS10}.
According to Lemma 1, we write for the first layer of the bilayer
LDPC convolutional code,
\begin{equation*}
p_1^{(i)}=\epsilon (q_1^{(i-1)})^{l_1-1} (q_2^{(i-1)})^{l_2} \hspace{2mm} \mathrm{and} \hspace{2mm} q_1^{(i)}=1-(1-p_1^{(i)})^{r_1-1}.
\end{equation*}
For the second layer of the code, we have
\begin{equation*}
p_2^{(i)}=\epsilon (q_1^{(i-1)})^{l_1} (q_2^{(i-1)})^{l_2-1} \hspace{2mm} \mathrm{and} \hspace{2mm} q_2^{(i)}=1-(1-p_2^{(i)})^{r_2-1}.
\end{equation*}

Since $r_1=r_2=r$ and $p_1^{(1)}=p_2^{(1)}=\epsilon$, we obtain
from the iterations $p_1^{(i)}=p_2^{(i)}$. Then the recursion can
be written as
\begin{equation*}
p^{(i)}=\epsilon(1-(1-p^{(i-1)})^{r-1})^{l_1+l_2-1}.
\end{equation*}
This indicates that the bilayer LDPC convolutional code has the
same BP threshold as the $\{l_1+l_2,r,L,w\}$ ensemble.

The rate of the bilayer LDPC convolutional code satisfies
\vspace{-1mm}
\begin{equation}
\label{BLDPCCRate}
\lim_{w\rightarrow\infty}\lim_{L\rightarrow\infty}R(l_1,l_2,r,r,L,w)=1-\frac{l_1}{r}-\frac{l_2}{r},
\vspace{-1mm}
\end{equation}
and the decoding threshold achieves the Shannon limit,
\begin{equation}
\label{BLDPCCLimit}
\lim_{r\rightarrow\infty}\lim_{w\rightarrow\infty}\lim_{L\rightarrow\infty}\epsilon^{BP}(l_1,l_2,r,r,L,w)=\frac{l_1}{r}+\frac{l_2}{r}.
\end{equation}

According to (\ref{Rate}) and (\ref{Capacity}), obviously the
single-layer code $\{l_1+l_2,r,L,w\}$ has the same rate as in
(\ref{BLDPCCRate}) and achieves the same limit as in
(\ref{BLDPCCLimit}). Therefore, the theorem is proven.
\end{proof}

For the design of bilayer LDPC convolutional codes in relay
channels, firstly we choose an $\{l_1,r,L,w\}$ ensemble which is
capacity achieving for the source-relay link. Afterwards the relay
generates the syndrome bits according to $\{l_2,r,L,w\}$ and
forwards them to the destination. The overall code structure is
consequently $\{l_1,l_2,r,r,L,w\}$. In the following we will show
this overall code is capacity achieving for the source-destination
link. In addition, it enables the highest achievable rate $R_{DF}$
of the relay channel.

\begin{theorem}
For a binary erasure relay channel, we can find an LDPC
convolutional code $\mathcal{C}_1=\{l_1,r,L,w\}$ achieving the
capacity for the source-relay link and simultaneously its bilayer
extension $\mathcal{C}=\{l_1,l_2,r,r,L,w\}$ achieving the capacity
for the source-destination link. Meanwhile, the above code
construction provides the highest achievable rate with
decode-and-forward relaying as in (\ref{AchiRate}).
\end{theorem}

\begin{proof}
We assume that the erasure probability for the source-relay link
and the source-destination link are $\epsilon_{SR}$ and
$\epsilon_{SD}$, respectively, and $\epsilon_{SR}<\epsilon_{SD}$.
The corresponding channel capacities for these two links are
\begin{equation*}
C_{SR}=1-\epsilon_{SR},\  C_{SD}=1-\epsilon_{SD}.
\end{equation*}

We use a regular LDPC convolutional code $\{l_1,r,L,w\}$ with
$l_1/r=\epsilon_{SR}$ for the transmission in the first phase.
According to (\ref{Rate}) and (\ref{Capacity}), we have
\vspace{-1mm}
\begin{equation*}
\lim_{w\rightarrow\infty}\lim_{L\rightarrow\infty}R(l_1,r,L,w)=1-\frac{l_1}{r}=1-\epsilon_{SR}
\end{equation*}\vspace{-2mm}
\begin{equation*}
\mathrm{and}\hspace{3mm}\lim_{r\rightarrow\infty}\lim_{w\rightarrow\infty}\lim_{L\rightarrow\infty}\epsilon^{BP}(l_1,r,L,w)=\epsilon_{SR}.
\end{equation*}
Hence, $\mathcal{C}_1$ is capacity achieving, and error-free
decoding can be guaranteed at the relay.

We assume that the number of variable nodes of $\mathcal{C}_1$ is
$N_V$ and the number of check nodes of $\mathcal{C}_1$ is
$N_{C1}$, then
\begin{equation*}
N_{C1}=l_1 N_V/r.
\end{equation*}
The number of additional bits needed at the destination is
\begin{equation*}
N_{C2}=N_V(C_{SR}-C_{SD})=N_V(\epsilon_{SD}-\epsilon_{SR}),
\end{equation*}
and these bits are provided by the syndrome generated at the
relay. At the destination, the total number of check nodes is
\begin{equation*}
N_C=N_{C1}+N_{C2}=N_V(\frac{l_1}{r}+\epsilon_{SD}-\epsilon_{SR})=\epsilon_{SD}N_V.
\end{equation*}
The additional $N_{C2}$ check equations bring in $r N_{C2}$ edges,
and the corresponding variable degree $l_2$ follows as
\begin{equation*}
l_2=r N_{C2}/N_V=r(\epsilon_{SD}-\epsilon_{SR}).
\end{equation*}
From Theorem 2, we have for the source-destination link
\begin{equation*}
\lim_{w\rightarrow\infty}\lim_{L\rightarrow\infty}R(l_1,l_2,r,r,L,w)=1-\frac{l_1+l_2}{r}=1-\epsilon_{SD},
\end{equation*}\vspace{-2mm}
\begin{equation*}
\mathrm{and}\hspace{3mm}\lim_{r\rightarrow\infty}\lim_{w\rightarrow\infty}\lim_{L\rightarrow\infty}\epsilon^{BP}(l_1,l_2,r,r,L,w)=\epsilon_{SD}.
\end{equation*}
Therefore, the overall code $\mathcal{C}$ achieves the capacity of
the source-destination link.

The number of channel uses in the first phase is $N_1=N_V$. In the
second phase, we can use another capacity-achieving LDPC
convolutional code to transmit the $N_{C2}$ syndrome bits to the
destination. Therefore, $N_2=N_{C2}/C_{RD}$ channel uses are
needed. The fraction
\begin{equation*}
\alpha'=\frac{N_1}{N_1+N_2} = \frac{C_{RD}}{C_{RD}+C_{SR}-C_{SD}}
\end{equation*}
equals the one in (\ref{fraction}), which maximizes the achievable
rate.
\end{proof}

From Theorem 3, we can conclude that the proposed regular bilayer
LDPC convolutional codes significantly simplify the code
optimization. Appropriate variable and check node degrees can
easily be computed from the parameters of the channels, and a
complicated optimization of irregular degree distributions as for
example in \cite{RY07} can be avoided.

\section{Numerical Results}
\label{sec:Results}

In this section, we firstly give numerical results for bilayer
LDPC convolutional code ensembles $\{l,r,L\}$ in binary erasure
relay channels. The source broadcasts its information bits with an
$\{l_1\!\!=\!\!3,r\!\!=\!\!10,L\!\!=\!\!100\}$ LDPC convolutional
code. At each time instant, the number of variable nodes is set to
be $M\!=\!2000$. At the relay, different values of $l_2$
($l_2\in\{2,3,4,5\}$) are chosen. Consequently, bilayer LDPC
convolutional codes of different rates are constructed. Note that
rate loss is inevitable for finite $L$ \cite{LSCZ10}. We compare
the decoding thresholds of both the single-layer code and the
bilayer codes with the corresponding Shannon limits. It can be
seen from Figure~\ref{fig:BilayerBECBER} that the gaps in between
are impressively small.

\begin{figure}[htb]
\vspace{-3.5mm}
\centering
  \includegraphics[width=0.83\columnwidth]{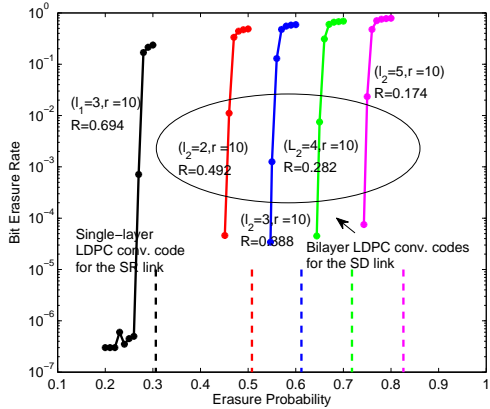}
  \vspace{-3mm}
  \caption{Bit erasure rate of bilayer LDPC convolutional codes with different overall rates in BECs.
  The solid curves show the simulation results, and the dashed lines indicate the Shannon limits.}
  \vspace{-2.5mm}
  \label{fig:BilayerBECBER}
\end{figure}

To evaluate the proposed bilayer LDPC convolutional codes under
more practical conditions, Figure~\ref{fig:BilayerAWGNBER} shows
the bit-error-rate performance for the AWGN channel. For
comparison purpose, we also include a regular bilayer LDPC block
code. For both types of the codes, we set $l_1=3$, $l_2=2$, and
$r=10$, which leads to approximately $R_{SR}=0.7$ and
$R_{SD}=0.5$. In addition, the lengths of both codes are chosen in
the way that the same hardware complexity \cite{PFSLZC08} is
needed. It can be observed that the bilayer LDPC convolutional
code clearly outperforms its block code counterpart.
Signal-to-noise ratio (SNR) gains of $0.5$~dB and $1.3$~dB are
obtained at the relay and at the destination, respectively.

\begin{figure}[htb]
\vspace{-2mm}
\centering
  \includegraphics[width=0.83\columnwidth]{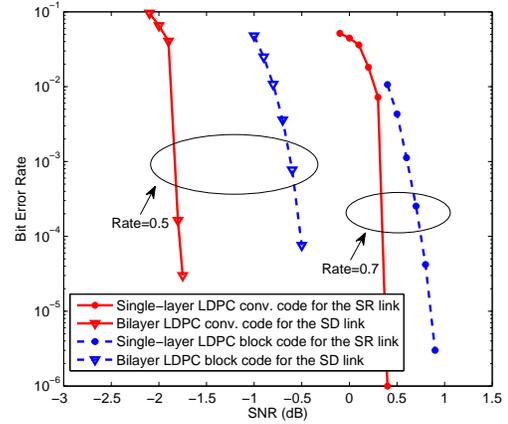}
  \vspace{-3mm}
    \caption{Comparison of bit error rate between a bilayer LDPC convolutional code and a regular bilayer LDPC block code in AWGN channels.}
    \vspace{-3mm}
  \label{fig:BilayerAWGNBER}
\end{figure}

\section{Conclusions}
\label{sec:Conclusion}

In this paper bilayer LDPC convolutional codes were proposed for
three-node relay channels. For a binary erasure relay channel, we
can find a bilayer LDPC convolutional code which is able to
simultaneously achieve the capacities of the source-relay link and
the source-destination link. Meanwhile, this code provides the
highest possible transmission rate with decode-and-forward
relaying. Moreover, the regular code structure significantly
reduces the complexity by avoiding the optimization of irregular
degree distributions. Numerical results were provided in both
binary erasure channels and AWGN channels. In binary erasure
channels, we can observe that the decoding thresholds are very
close to the Shannon limits. In AWGN channels, a significant gain
in terms of SNR is achieved compared with its block code
counterpart.

\small
\bibliographystyle{IEEEbib}
\bibliography{bibnames,references}

\end{document}